%% file: approximal_operator_and_audio_inpainting.tex
\documentclass[preprint]{elsarticle}
\usepackage[utf8]{inputenc}
\usepackage{csquotes}
\usepackage{amsmath}
\usepackage{amsfonts}
\usepackage{amssymb}
\usepackage{amsthm}
\usepackage{graphicx}
\usepackage{hyperref}
\usepackage{subfig}
\usepackage{geometry}
\usepackage[linesnumbered,algoruled,czech]{algorithm2e}
\usepackage[super]{nth}
\usepackage[dvipsnames,table]{xcolor}
\usepackage{mathtools}

\include{defs}

\begin{document}
	
\title{Approximal operator with application to audio inpainting}

\author{Ondřej Mokrý}
\ead{170583@vutbr.cz}
\author{Pavel Rajmic\corref{cor1}}
\ead{rajmic@feec.vutbr.cz}

\address{Signal Processing Laboratory,
	Department of Telecommunications,\\
	Faculty of Electrical Engineering and Communication,
	Brno University of Technology,
	Technická 12,
	612 00 Brno,
	Czech Republic}

\cortext[cor1]{Corresponding author.
The work was supported by the Czech Science Foundation (GA\v{C}R) project number 20-29009S.
The authors would like to thank Rémi Gribonval for sharing his expertise and discussing the manuscript.}

\begin{abstract}
In their recent evaluation of time-frequency representations and structured sparsity approaches to audio inpainting,
Lieb and Stark (2018) have used a~particular mapping as a~proximal operator.
This operator serves as the fundamental part of an iterative numerical solver.
However, their mapping is improperly justified.
The present article proves that their mapping is indeed a proximal operator,
and also derives its proper counterpart.
Furthermore, it is rationalized that Lieb and Stark's operator can be understood as an approximation of the proper mapping.
Surprisingly, in most cases, such an approximation
(referred to as the approximal operator)
is shown to provide even better numerical results in audio inpainting compared to its proper counterpart, while being computationally much more effective.

\end{abstract}


\maketitle

\section{Introduction}

	In signal and image processing, the so-called proximal splitting algorithms represent an effective way of finding numerical solutions to various problems \cite{combettes2011proximal}.
	However, despite their conceptual simplicity, proximal algorithms are not always computationally tractable.
	For instance,
	in the area of audio signal restoration,
	it is often necessary to handle the proximal operator of a composition of a~linear (or affine) transform and a~functional.
	This happens in many variations of audio inpainting tasks
	\cite{javevr86,Adler2012:Audio.inpainting,Lieb2018:Audio.Inpainting}
	(further discussed in Sec.\,\ref{sec:experiments}),
	audio declipping \cite{SiedenburgKowalskiDoerfler2014:Audio.declip.social.sparsity, ZaviskaRajmicSchimmel2019:Psychoacoustics.l1.declipping, Kitic2013:Consistent.iter.hard.thresholding, Kitic2015:Sparsity.cosparsity.declipping},
	and audio dequantization \cite{RenckerBachWangPlumbley2018:Fast.iterative.shrinkage.declip.dequant-iTwist18, BrauerGerkmannLorenz2016:Sparse.reconstruction.of.quantized.speech.signals,ZaviskaRajmic2020:Dequantization}.
	
	To introduce the concept more specifically, let $L\colon\mathbb{W}\to\mathbb{V}$ be a linear mapping between two vector spaces, and let $g\colon\mathbb{V}\to\RR$ be a convex lower semi-continuous functional.
	We are interested in computing the proximal operator $\prox_f$,
	where $f=g\circ L$ (the symbol $\circ$ shall denote the composition of two functions),
	i.e.\ the mapping $L$ is applied first, followed by $g$.
	
	In some cases, explicit formulas for $\prox_{g\circ L}$ are available.
	For instance, \cite{combettes2011proximal} provides an explicit formula for the case of real-valued operator $L$
	between finite-dimensional spaces such that $LL^\top=\alpha\Id$.
	A similar result is presented in \cite{RajmicZaviskaVeselyMokry2019:Axioms}, where $L$ is assumed to be complex-valued and satisfying the condition that $L\adjoint{L}$ is diagonal
	(the asterisk denotes the adjoint operator).
	However, it is limited to the case when $\prox_{g}$ is the operator of projection onto a~box-type set.
	
	In the present contribution, we provide an analysis of the case when $L^\top L=\alpha\Id$, which is closely related to
	\cite{combettes2011proximal},
	and this is motivated by a recent signal processing application
	\cite{Lieb2018:Audio.Inpainting}.
	In Sec.\,\ref{sec:theory}, we derive a formula for $\prox_{g\circ L}$ in such a~scenario.
	Furthermore, an approximation of the derived proximal operator is introduced and analyzed,
	provided that an explicit formula for $\prox_{\alpha g}$ exists.
	Sec.\,\ref{sec:experiments} shows its usefulness in the case of sparsity-based audio inpainting.
	The error arising from computing the proximal step of an iterative algorithm only approximately
	is evaluated on the example of audio inpainting.
		
	Throughout the whole article, the symbol $\norm{\cdot}$ may denote a general norm on a Hilbert space $\mathbb{V}$,
	the common $\ell_2$ norm on a~finite-dimensional real or complex vector space,
	or the induced operator norm.
	The particular case should be clear from the context.
	The inner product inducing $\norm{\cdot}$ will be denoted $\langle\cdot,\cdot\rangle$.
	Should any other (pseudo)norm appear, it will be identified using the lower index notation, e.g.\ $\norm{\cdot}_0$, $\norm{\cdot}_1$.

\section{Proximal operator of a composition of a proper convex function and an affine mapping and its approximation}
\label{sec:theory}

\subsection{Theoretical proposition}

	It has already been stated that the novel proposition of the paper is related to the known formula for $\prox_{g\circ L}$ in the case of semi-orthogonal $L$, i.e. $LL^\top =\alpha\Id$.
	To build upon this relation, we start by quoting the corresponding lemma.
	
	\begin{lemma}[the original lemma from {\cite[p.~140]{Beck2017:First.Order.Methods}}]
		\label{lemma:the.original}
		Let $g\colon\RR^m\to\RR\cup\{\infty\}$ be a~proper convex function, and let $f(\x) = g(L\x+\b)$, where $\b\in\RR^m$ and $L\colon\mathbb{V}\to\RR^m$ is a~linear transformation satisfying $L\transp{L}=\alpha\Id$ for some constant $\alpha>0$.
		Then for any $\x\in\mathbb{V}$,
		\begin{equation}
		\prox_f(\x) = \x + \alpha^{-1}\transp{L}\left(\prox_{\alpha g}(L\x+\b)-L\x-\b\right).
		\label{eq:the.original}
		\end{equation}
	\end{lemma}

	Note that when $\prox_{\alpha g}$ is known explicitly, Eq.\,\eqref{eq:the.original} provides an explicit form of $\prox_f$.
	As an example, take $f=\norm{L\cdot}_1=\norm{\cdot}_1\circ L$.
	Then $g$ is the $\ell_1$ norm and the corresponding $\prox_{\alpha g}$ is the soft thresholding operator with the threshold $\alpha$.
	Lemma \ref{lemma:the.original} may also be used when $g$ is the indicator function of a~closed convex set $C\subset \RR^m$,
	\begin{equation}
	g(\x) = \iota_C (\x) = \begin{cases}
	0\quad &\x\in C,\\
	\infty\quad &\x\notin C.
	\end{cases}
	\end{equation}
	In such a case, the proximal operator of $\alpha g$ is the operator of projecting onto $C$, denoted
	$\prox_{\alpha g}(\x) = \proj_C(\x)$.
	Additionally, the frame theory provides an example of $L$ that fits the lemma---it may be the synthesis operator of a~tight frame \edit{in $\RR^m$} \cite{christensen2008,christensen2003,Grochenig2001:Foundations.T-F.analysis,Balazs2011:nonstatgab,Necciari2013:ERBlet}.
	\edit{In brief, a frame is a (commonly overcomplete) set of vectors generating $\RR^m$;
	this allows representation of any vector from $\RR^m$ using a~sequence of coefficients in $\mathbb{V}=\RR^n$.
	These coefficients are obtained by the injective analysis operator.
	Its adjoint, the synthesis operator, produces a vector in $\RR^m$ as a linear combination of the generators.
	In the setting of Lemma \ref{lemma:the.original}, $L$ may represent the synthesis operator and $\transp{L}$ the analysis operator.
	Tight frames are a~special subgroup of frames, for which the condition $L\transp{L}=\alpha\Id$ holds.}
	
	The question is:
	If we want to employ the analysis operator \edit{in place of $L$}, how will formula \eqref{eq:the.original} be affected?
	The following lemma answers the question.
	
	\begin{lemma}[the proposed lemma]
		\label{lemma:the.new}
		Let $g\colon\RR^m\to\RR\cup\{\infty\}$ be a proper convex function, and let $f(\x) = g(\ana\x+\b)$, where $\b\in\RR^m$ and $\ana\colon\mathbb{V}\to\RR^m$ is a~linear transformation satisfying $\syn\!\ana=\alpha\Id$ for some constant $\alpha>0$.
		Then for any $\x\in\mathbb{V}$,
		\begin{equation}
		\prox_f(\x) = \alpha^{-1}\syn\left(\prox_{\alpha g + \irab}(\ana\x+\b)-\b\right),
		\label{eq:the.new}
		\end{equation}
		where $\irab$ is the indicator function of the affine space which we obtain by shifting the range space of $\ana$ by a~vector $\b$.
	\end{lemma}

	Although the proof is given in \ref{appendix:proofs} including the vector $\b$,
	we will for simplicity further assume $\b = \mathbf{0}$.
	
	From the viewpoint of frame theory, the operator $\ana$ in Lemma \ref{lemma:the.new} is the analysis operator of a tight frame.
	Suppose that $\mathbb{V} = \RR^n$, i.e. $\ana\colon\RR^n\to\RR^m$.
	It is a straightforward consequence of the assumption $\syn\!\ana = \alpha\Id$ that in such a case,
	it must hold $n\leq m$ due to the property
	\begin{equation}
		n = \rank{\ana^\top\ana} = \rank{\ana} \leq m.
		\label{eq:rank}
	\end{equation} 
	
	Observe that the crucial difference between the two lemmas is that in the latter case,
	the image of the proximal operator of $\alpha g$ is forced to lie in the subspace $\RA$.
	Omitting this condition, Eq.\,\eqref{eq:the.new} would be obtained simply by using Eq.\,\eqref{eq:the.original} and plugging in the property $\syn\!\ana = \alpha\Id$.
	Such a step is correct when $\ana$ is surjective, in which case $m = n$ (as a~consequence of Eq.\,\eqref{eq:rank}) and the indicator function $\iota_{\RA}$ equals zero on whole $\RR^n$ and therefore $\prox_{\alpha g + \iota_{\RA}} = \prox_{\alpha g}$.
	However, the consequence of  $\ana$ being full-rank\footnote{%
	In the case of $\ana\colon\RR^n\to\RR^m$, the operator $\ana$ has full rank if \edit{$\rank{\ana} = \min\{m,n\}$}.}
	is that the operator is unitary and by taking $\ana = \transp{L} = L^{-1}$, Lemmas \ref{lemma:the.original} and \ref{lemma:the.new} coincide.
	
	Finally, observe that Lemma \ref{lemma:the.original} provides a~constructive way to compute
	$\prox_{g\circ L}$ when $\prox_{\alpha g}$ is known explicitly.
	On the contrary, Lemma \ref{lemma:the.new} is not constructive in the case of $n < m$,
	since the explicit form of $\prox_{\alpha g}$ does not guarantee the existence of an explicit form of $\prox_{\alpha g + \iota_{\RA}}$.

\subsection{Explicit approximation}

	The strength of Lemma \ref{lemma:the.original} is that it offers an explicit form of $\prox_f$ when
	the explicit form of $\prox_{\alpha g}$ is available.
	This is possible for example in the aforementioned cases of the (weighted) $\ell_1$ norm or the indicator function in the role of $g$.
	However, the proximal operator including an additional restriction, which is the case of $\prox_{\alpha g + \iota_{\RA}}$ in Lemma \ref{lemma:the.new}, is seldom known%
	\footnote{A~rare example is the proximal operator of (weighted) $\ell_1$ norm over a box \cite[pp.~145--146]{Beck2017:First.Order.Methods}.},
	resulting in the need for a~reasonable approximation of $\prox_{\alpha g + \iota_{\RA}}$.
	
	Since
	the orthogonal projection onto $\RA$ in the case of $\syn\!\ana=\alpha\Id$ is expressed easily as \cite{christensen2008}
	\begin{equation}
		\proj_{\RA} (\x) = \alpha^{-1}\ana\syn\x,
		\label{eq:proj.RA}
	\end{equation}
	a~natural possibility is to study the approximation
	\begin{equation}
		\prox_{\alpha g + \iota_{\RA}} \approx \proj_{\RA}\circ\prox_{\alpha g},
	\end{equation}
	i.e.\ the composition of two known operators.
	Such an approximation will thus take the form of 
	\begin{equation}
		\badprox_{f}(\x) = \alpha^{-1}\syn\left(\proj_{\RA}\left(\prox_{\alpha g}(\ana\x)\right)\right),
		\label{eq:bad.prox}
	\end{equation}
	where we introduced the denotation `$\badprox$' for the \emph{approximal operator}.

	The most important property of $\badprox_f$ is pointed out in the following lemma.
	\begin{lemma}
		Under the conditions of Lemma \ref{lemma:the.new},
		the approximal operator is
		\begin{equation}
			\badprox_{f}(\x) = \alpha^{-1}\syn\prox_{\alpha g}(\ana\x)
			\label{eq:shortened.bad.prox}
		\end{equation}
		and it is a~proximal operator of some convex lower semi-continuous function.
		\label{lemma:approx.is.prox}
	\end{lemma}

	Note that the shortened form of $\badprox_f$ in Eq.\,\eqref{eq:shortened.bad.prox} was obtained by plugging \eqref{eq:proj.RA} into \eqref{eq:bad.prox} and using the above-mentioned property $\syn\!\ana=\alpha\Id$.
	
	The lemma is proved in \ref{appendix:proofs}, justifying two properties of proximal operators,
	one of which is the non-expansivity of $\badprox_{f}$.
	As mentioned in \cite[p.\,7]{Gribonval2018:Characterization.of.prox},
	the non-expansivity plays a~role in the convergence analysis of iterative proximal algorithms.

	\edit{Recently, a~proposition similar to Lemma \ref{lemma:approx.is.prox} has been independently proven in
	\cite[Theorem 3.3]{Hasannasab2020:Parseval.Proximal.Neural.Networks}.
	Furthermore, \cite[Theorem 3.4]{Hasannasab2020:Parseval.Proximal.Neural.Networks} proposes a~formula for the function $\phi$ such that $\badprox_f=\prox_\phi$, which is summarized in the following lemma.
	\begin{lemma}[{adapted based on \cite[Theorem 3.4]{Hasannasab2020:Parseval.Proximal.Neural.Networks}}]
		Denote $g\conv h$ the infimal convolution of the functions $g,h\colon\RR^m\to\RR\cup\{\infty\}$, defined pointwise as
		\begin{equation}
			(g\conv h)(\x) = \inf_{\y\in\RR^m}\left\{ g(\y) + h(\x-\y) \right\}\quad\forall\x\in\RR^m.
			\label{eq:infimal.convolution}
		\end{equation}
		The operator $\badprox_{f}$ defined by Eq.\,\eqref{eq:shortened.bad.prox}, assuming $\syn\!\ana=\alpha\Id$, is the proximal operator of the function $\phi\colon\mathbb{V}\to\RR\cup\{\infty\}$ defined pointwise as
		\begin{equation}
			\begin{split}	
			\phi(\x) &= \frac{1}{\alpha}\left[\alpha g\conv \left(\tfrac{1}{2}\norm{\cdot}^2 + \iota_{\mathcal{N}(\syn)}\right)(\ana\x)\right] \\
			&= \alpha^{-1}\inf_{\y\in\RR^m}\left\{ \alpha g(\y) + \frac{1}{2}\norm{\ana\x-\y}^2 + \iota_{\mathcal{N}(\syn)}(\ana\x-\y) \right\},
			\end{split}
			\label{eq:phi}
		\end{equation}
		where $\mathcal{N}(\syn)\subset\RR^m$ denotes the null space of the operator $\syn\colon\RR^m\to\mathbb{V}$.
		\label{lemma:phi}
	\end{lemma}
	It is straightforward to derive the lemma from \cite[Theorem 3.4]{Hasannasab2020:Parseval.Proximal.Neural.Networks}, using \cite[Lemma 3.1]{Hasannasab2020:Parseval.Proximal.Neural.Networks} and the explicit form of the pseudoinverse of $\ana$,
	i.e.\ $\pinv{\ana} = \inv{(\syn\!\ana)}\syn = \alpha^{-1}\syn$.
	Although Eq.\,\eqref{eq:phi} provides a compact definition for the function $\phi$, it still includes the infimal convolution, impractical in numerical applications.
	A~minor insight into the properties of $\phi$ is provided by the following lemma.
	\begin{lemma}
		\label{lemma:leq}
		For any $\x\in\mathbb{V}$, $\phi(\x)\leq f(\x)$.
	\end{lemma}
	\begin{proof}
		From Eq.\,\eqref{eq:phi}, it holds
		(putting $\y=\ana\x$)
		\begin{equation}
			\phi(\x) \leq \alpha^{-1}\left\{ \alpha g(\ana\x) + \frac{1}{2}\norm{\ana\x-\ana\x}^2 + \iota_{\mathcal{N}(\syn)}(\ana\x-\ana\x) \right\} = g(\ana\x) = f(\x)
		\end{equation}
		for any $\x\in\mathbb{V}$.
	\end{proof}
	Besides this theoretical result, an illustrative example of $\phi$ for $g=\norm{\cdot}_1$ is included in \ref{appendix:example}.}
	In the following section, an application from the field of audio processing demonstrates the suitability of $\badprox_f$ as an approximation of $\prox_f$.

\section{Experiments}
\label{sec:experiments}

\subsection{Sparsity-based audio inpainting}

	Audio inpainting is a rather modern term for the task of filling highly degraded or missing samples of digital audio \cite{Adler2012:Audio.inpainting}.
	The same task is referred to as the interpolation of missing samples \cite{javevr86, Etter1996:Interpolation_AR} or packet-loss concealment \cite{Lindblom2002:PLC.sinusoidal, Rodbro.Packet.Loss.Concealment}.
	
	Popular audio inpainting methods are based on the assumption that musical audio signals are sparse with respect to a suitable time-frequency (TF) transform.
	To proceed with the formalization of this assumption, denote $\sig\in\CC^n$ the observed (i.e.\ degraded) signal
	and let
	$\mask\colon\CC^n\to\CC^n$ be the operator which fills with zeros the missing information in the signal.
	Denote $\Gamma$ the set of signals consistent with the observed signal $\sig$,
	\begin{equation}
		\Gamma = \left\{ \x\in\CC^n\mid\mask\x = \mask\y \right\}.
	\end{equation}
	The inpainting task is then formulated as the following optimization task:
	\begin{center}
		\emph{Find the signal from $\Gamma$ with the corresponding TF coefficients as sparse as possible.}
	\end{center}	
	Let us recall two operators:
	Let $\ana\colon\CC^n\to\CC^m$, $m\geq n$, be the analysis operator, expanding the time-domain signal into the vector of TF coefficients,
	and let $\adjoint{\ana}\colon\CC^m\to\CC^n$ be its synthesis counterpart, producing a signal, given the coefficients.
	With this notation, the above formulation can be understood in two ways:%
	\begin{subequations}\label{eq:formulation}
		\begin{align}
			&\argmin_\z\;\norm{\z}_0\quad\text{s.\,t.}\quad\adjoint{\ana}\z\in\Gamma,
			\label{eq:formulation.syn}\\			
			&\argmin_\x\;\norm{\ana\x}_0\quad\text{s.\,t.}\quad\x\in\Gamma.
			\label{eq:formulation.ana}
		\end{align}
	\end{subequations}
	The symbol $\norm{\cdot}_0$ denotes sparsity, i.e.\ the pseudonorm that counts the non-zero entries of the argument.
	Since Eq.\,\eqref{eq:formulation.syn} includes the synthesis operator, it is referred to as the \emph{synthesis formulation};
	by the same reasoning, Eq.\,\eqref{eq:formulation.ana} is called the \emph{analysis formulation} \cite{Elad05analysisversus}.

	Note that the two formulations are equivalent only when the operator $\ana$ is unitary,
	i.e.\ it holds $\adjoint{\ana} = \ana^{-1}$.
	However, this is not the case of the commonly used Gabor, wavelet or ERBlet transforms \cite{christensen2003, Necciari2013:ERBlet}.
	
\subsubsection{Convex relaxation}
	
	Formulations \eqref{eq:formulation} are problematic in that they include the $\ell_0$ pseudonorm,
	resulting in the task being NP-hard.
	Two possible classes of methods exist to solve \eqref{eq:formulation}, in general only approximately.
	Either a~non-convex heuristic algorithm is employed to tackle \eqref{eq:formulation},
	e.g.\ the OMP \cite{Adler2012:Audio.inpainting} or SPAIN \cite{MokryZaviskaRajmicVesely2019:SPAIN}.
	Alternatively, the task needs to be relaxed to a convex optimization problem \cite{Bruckstein.etc.2009.SIAMReviewArticle,DonohoElad2003:Optimally}.
	Denoting $\sparse\colon\CC^m\to\RR$ a convex sparsity-related penalty, the convex relaxations of \eqref{eq:formulation.syn} and \eqref{eq:formulation.ana} attain the form
	\begin{subequations}\label{eq:convexformulation}
		\begin{align}
		&\argmin_\z\;\left\{\sparse(\z)+\iota_\Gamma(\adjoint{\ana}\z)\right\},
		\label{eq:convexformulation.syn}\\			
		&\argmin_\x\;\left\{\sparse(\ana\x)+\iota_\Gamma(\x)\right\}.
		\label{eq:convexformulation.ana}
		\end{align}
	\end{subequations}

	The constraints from \eqref{eq:formulation} are now included in the objective function
	using the indicator function $\iota_\Gamma$.
	Note that $\Gamma$ is a convex set by design.
	Thus, in order to have an overall convex problem, $\sparse$ has to be a~convex function, and it shall promote sparsity.
	As an example of a suitable and widely used penalty $\sparse$, we mention the (weighted) $\ell_1$ norm.
	
\subsubsection{Inconsistent reformulation}

	In \eqref{eq:convexformulation}, the solution is forced to be equal to the observed signal in its reliable (non-degraded) parts.
	However, this assumption may be too strong, for instance when the observed signal $\sig$ is noisy.
	In such a case, the alternative reformulations to solve are the so-called inconsistent problems
	\begin{subequations}\label{eq:inconsistentformulation}
		\begin{align}
		&\argmin_\z\;\left\{\sparse(\z)+\lambda\norm{\mask\adjoint{\ana}\z-\mask\sig}\right\},
		\label{eq:inconsistentformulation.syn}\\			
		&\argmin_\x\;\left\{\sparse(\ana\x)+\lambda\norm{\mask\x-\mask\sig}\right\},
		\label{eq:inconsistentformulation.ana}
		\end{align}
	\end{subequations}
	where the parameter $\lambda>0$ balances consistency with the data and the sparsity.

\subsubsection{Solving the task}
	
	Formulations \eqref{eq:convexformulation} and \eqref{eq:inconsistentformulation} consist of sums of lower semicontinuous convex functions.
	Such a scenario allows using proximal algorithms to solve the tasks numerically
	\cite{combettes2011proximal}.
	To find the minimum of the sum of lower semicontinuous convex functions
	$f_1$ and $f_2$
	using the proximal splitting approach, one must be able to evaluate both $\prox_{f_1}$ and $\prox_{f_2}$, or, in the case of differentiable $f_1$ or $f_2$, the corresponding gradient.
	These are evaluated in every iteration, meaning that simple, computationally cheap formulas are preferred.
	
	To choose suitable algorithms for solving \eqref{eq:convexformulation} and \eqref{eq:inconsistentformulation},
	assume that the explicit form of $\prox_\sparse$ is available.
	Note that in practice, this assumption is not too restrictive.
	Sometimes the situation is even the opposite, meaning that a suitable operator in the place of $\prox_\sparse$ is used although an explicit form of $\sparse$ is not available, for instance in the case of the (persistent) empirical Wiener and similar operators \cite{Gribonval2018:Characterization.of.prox,kowalski2012social}.
	Moreover, assume using a~tight frame as the TF transformation,
	i.e.\ $\adjoint{\ana}\ana=\alpha\Id$.
	To solve the synthesis model \eqref{eq:convexformulation.syn}, Lemma~\ref{lemma:the.original} is used (putting $L=\adjoint{\ana}$) to compute
	the prox of the second term,
	\begin{equation}
		\prox_{\iota_\Gamma\circ\adjoint{\ana}}(\x) = \x + \alpha^{-1}\ana\left( \proj_\Gamma(\adjoint{\ana}\x) -\adjoint{\ana}\x \right),
	\end{equation} 
	which enables the use of the Douglas--Rachford (DR) algorithm \cite[Sec.\,IV]{combettes2011proximal}.
	The reason is that for any time-domain signal $\vect{s}\in\CC^n$, the projection $\proj_\Gamma(\vect{s})$ is computed
	simply entry-by-entry by setting the samples at the reliable positions equal to the observed ones, and keeping the rest unchanged.
	
	It should be pointed out that Lemma \ref{lemma:the.original} is used here to handle a~complex-valued operator.
	This is treated in \ref{appendix:complex}, where one can find the argumentation
	why using complex variables in place of real variables does not affect the proposed theoretical results of both lemmas.
	
	To solve the analysis model \eqref{eq:convexformulation.ana},
	$\prox_{\sparse\circ\ana}$ has to be known in order to be used in the DR algorithm;
	formally, it is  constructed using Lemma \ref{lemma:the.new} as
	\begin{equation}
		\prox_{\sparse\circ\ana}(\x) = \alpha^{-1}\adjoint{\ana}\prox_{\alpha \sparse + \iota_{\RA}}(\ana\x).
		\label{eq:good.prox}
	\end{equation}
	The numerical treatment of \eqref{eq:good.prox}, however, requires choosing one of the following three options:
	\begin{itemize}
		\item
			$\prox_{\alpha \sparse + \iota_{\RA}}$ is only approximated by the operator $\badprox_{\sparse\circ\ana}$ defined by Eq.\,\eqref{eq:shortened.bad.prox},
		\item
			$\prox_{\sparse\circ\ana}(\x) = \argmin_{\u}\left\{ \sparse(\ana\u) +\frac{1}{2}\norm{\u - \x}^2 \right\}$ is computed using a~nested iterative subroutine,
			since it is by definition nothing but another optimization task,
		\item
			finally, the Chambolle--Pock (CP) algorithm \cite{ChambollePock2011:First-Order.Primal-Dual.Algorithm} can also be used instead of the DR algorithm
			to solve \eqref{eq:convexformulation.ana} without nested iterations,
			since it allows composing a~convex functional with a~linear mapping.
	\end{itemize}

	In analogy to the above,
	solving \eqref{eq:inconsistentformulation.syn} is not a problem, but for \eqref{eq:inconsistentformulation.ana} the approximation of $\prox_{\sparse\circ\ana}$ is needed.
	For both inconsistent cases, a suitable iterative algorithm to produce the numerical result is ISTA or FISTA (fast iterative shrinkage-thresholding algorithm) \cite{beck2009fast}.

\subsection{Evaluation setting}
	
	Our experiments basically follow and extend the experiment from \cite[Sec.\,3]{Lieb2018:Audio.Inpainting}.
	The motivation is that \cite{Lieb2018:Audio.Inpainting} applies Lemma \ref{lemma:the.original} instead of Lemma \ref{lemma:the.new} to compute $\prox_{\sparse\circ\ana}$.
	As a result, \cite{Lieb2018:Audio.Inpainting}
	actually uses $\badprox_{\sparse\circ\ana}$ (unintentionally).
	This is clear from plugging $\ana$ in place of $L$ in  Eq.\,\eqref{eq:the.original},
	which (recalling $\syn\!\ana=\alpha\Id$) directly produces $\badprox_{\sparse\circ\ana}$.
	The MATLAB codes and data from \cite{Lieb2018:Audio.Inpainting}
	are available at \url{https://github.com/flieb/AudioInpainting}.
	We recomputed their results using the shared code, and on top of that, we aimed at answering the question:
	\begin{center}
	\emph{Does the more accurate (but more complicated) use of\/ $\prox_{\sparse\circ\ana}$ produce better results\\ in the audio inpainting task, compared to the use of\/ $\badprox_{\sparse\circ\ana}$?}
	\end{center}
	Hence, our contribution concerns the analysis model.
	There we decided to use the CP algorithm instead of the DR algorithm in the consistent case,
	while in the inconsistent case,
	the CP algorithm approximates $\prox_{\sparse\circ\ana}$ in each iteration of FISTA.
	
	Exactly the same as in \cite{Lieb2018:Audio.Inpainting}, the algorithms are tested on four musical signals from \cite{siedenburg2013persistent} sampled at 44.1~kHz
	that are degraded by dropping out 80 \% of the samples at random positions.
	As the convex sparsity-related penalty, the $\ell_1$ norm is used, i.e.\ $\sparse(\cdot) = \norm{\cdot}_1$.
	Three different transforms \cite[Sec.\,3.3.1]{Lieb2018:Audio.Inpainting} are used:
		\begin{itemize}
			\item Gabor transform (GAB) with the Hann window of length 23\,ms (1024 samples), time sampling parameter $a = 3.6$\,ms (160 samples) and $M = 3125$ frequency channels,
			\item ERBlet transform (ERB) with $\mathit{qvar} = 0.08$ and $\mathit{bins} = 18$,
			\item wavelet transform (WAV) with $f_\mathrm{min} = 100$\,Hz, $\mathit{bw} = 3$\,Hz and $\mathit{bins} = 120$.
		\end{itemize}
	\edit{For the algorithm specific parameters, see \cite[Sec.\,3.3.3 and 4]{Lieb2018:Audio.Inpainting} and particularly the accompanying GitHub repository with the MATLAB codes (link in Sec.\,\ref{sec:software}).}

	The reconstruction performance is evaluated using the $\snr$, computed in agreement with \cite{Lieb2018:Audio.Inpainting} as
	\begin{equation}
	\snr(\vect{s}, \hat{\vect{s}}) = 20\log_{10}\frac{\sigma(\vect{s})}{\sigma(\vect{s} - \hat{\vect{s}})},
	\label{eq:snr}
	\end{equation}
	where $\vect{s}$ is the original (non-degraded) signal, $\hat{\vect{s}}$ is the reconstructed signal and $\sigma$ denotes the standard deviation.
	\edit{Since we work with real time-domain signals only, both $\vect{s}$ and $\hat{\vect{s}}$ in formula \eqref{eq:snr} are real.}
	It is worth noting that the reliable parts of $\vect{s}$ and $\hat{\vect{s}}$ are not taken into account in Eq.\,\eqref{eq:snr}.

\subsection{Results and discussion}

	Table~\ref{tab:florian} is a slightly reordered reproduction of the table provided in \cite[Sec.\,4]{Lieb2018:Audio.Inpainting}.
	As explained above, the results based on the synthesis model correspond to \eqref{eq:convexformulation.syn} and \eqref{eq:inconsistentformulation.syn},
	whereas the analysis-based results only approximate the solutions to Eq.\,\eqref{eq:convexformulation.ana} and \eqref{eq:inconsistentformulation.ana}, since the operator $\badprox_{\sparse\circ\ana}$ is used.
	
	\begin{table}[h]
		\centering
		\caption{Values of $\snr$ in dB from \cite[Table~2]{Lieb2018:Audio.Inpainting},
		based on the four test signals, three TF dictionaries
		and for both the consistent and the inconsistent approach, corresponding to problems \eqref{eq:convexformulation} and \eqref{eq:inconsistentformulation}, respectively.}
		\label{tab:florian}
		\begin{tabular}{|ll|rrr|rrr|}
			\hline
			\# & & \multicolumn{3}{l|}{DR (consistent)} & \multicolumn{3}{l|}{FISTA (inconsistent)} \\ \cline{3-8} 
			& & GAB  & WAV & ERB  & GAB & WAV & ERB \\ \hline
			1 & Synthesis & $18.7$ & $26.0$ & $26.4$ & $15.5$ & $25.5$ & $25.9$ \\
			  & Analysis  & $16.9$ & $25.9$ & $26.3$ & $18.6$ & $25.2$ & $25.6$ \\
			2 & Synthesis & $20.1$ & $25.9$ & $25.9$ & $16.8$ & $25.1$ & $25.2$ \\
			  & Analysis  & $18.3$ & $25.6$ & $25.6$ & $19.7$ & $25.1$ & $25.2$ \\
			3 & Synthesis & $18.6$ & $19.2$ & $19.3$ & $17.4$ & $18.9$ & $19.1$ \\
			  & Analysis  & $17.9$ & $19.2$ & $19.3$ & $18.5$ & $19.2$ & $19.2$ \\
			4 & Synthesis & $16.2$ & $19.8$ & $20.4$ & $13.6$ & $19.3$ & $20.1$ \\
			  & Analysis  & $15.1$ & $19.7$ & $20.4$ & $16.1$ & $19.7$ & $20.4$ \\ \hline
		\end{tabular}
	\end{table}

	Since our modifications stem from the use of $\badprox_{\sparse\circ\ana}$ in the analysis-based model, Table~\ref{tab:difference1} presents only the results of this model.
	Also, for better readability, only the difference between the values of $\snr$ is shown.
	\begin{table}[h]
		\centering
		\caption{Difference between the values of $\snr$ taken from the original experiment and from its correct implementation.
		The latter uses the CP algorithm instead of the DR algorithm in the consistent case, and in the inconsistent case, $\prox_{\sparse\circ\ana}$ is evaluated using the CP algorithm in each iteration of FISTA.
		Only the results of the analysis-based approach are shown.
		Positive values indicate cases in which the new implementation performs better.}
		\label{tab:difference1}
		\begin{tabular}{|l|rrr|rrr|}
			\hline
			\# & \multicolumn{3}{l|}{DR/CP (consistent)} & \multicolumn{3}{l|}{FISTA (inconsistent)} \\ \cline{2-7} 
			& GAB & WAV & ERB & GAB & WAV & ERB \\ \hline
			1 & $1.71$ & $ 0.09$ & $ 0.15$ & $0.02$ & $ 0.01$ & $ 0.01$ \\
			2 & $1.72$ & $-0.05$ & $-0.04$ & $0.11$ & $-0.02$ & $ 0.00$ \\
			3 & $0.68$ & $-0.02$ & $-0.01$ & $0.03$ & $ 0.00$ & $-0.02$ \\
			4 & $1.11$ & $ 0.07$ & $ 0.00$ & $0.02$ & $-0.01$ & $-0.01$ \\ \hline
		\end{tabular}		
	\end{table}

	Table~\ref{tab:difference1} shows two remarkable results.
	On the right side, the $\snr$ for FISTA is almost independent of whether
	$\prox_{\sparse\circ\ana}$ is computed accurately or not.
	The other note concerns the left side, showing the results of DR and CP algorithms, representing the use of $\badprox_{\sparse\circ\ana}$ and $\prox_{\sparse\circ\ana}$, respectively.
	Here, the more accurate approach with $\prox_{\sparse\circ\ana}$ outperforms the approximate approach with $\badprox_{\sparse\circ\ana}$ only with the Gabor dictionary,
	whereas for wavelets and ERBlets
	\edit{the results are hardly distinguishable in terms of the SNR.}
	
	Recall that the inconsistent analysis model employs a~nested iterative CP algorithm within FISTA,
	which results in the computational cost being remarkably higher compared to the synthesis model. 
	
	Furthermore, we examine the performance of the DR (original) and the CP (new) algorithm in the analysis model while choosing stricter convergence criteria than in the previously described experiment.
	In the implementation of \cite{Lieb2018:Audio.Inpainting}, the parameters indicating convergence are the maximum number of iterations (\texttt{param.maxit}) and the relative norm of solutions in subsequent iterations (\texttt{param.tol}).
	The algorithm stops when either of the criteria is reached.
	The settings of these parameters in the experiments are as follows:
	
	\begin{center}
		\begin{tabular}{lcc}
			criterion & setting for Tables \ref{tab:florian} and \ref{tab:difference1} & setting for Tables \ref{tab:difference2} and \ref{tab:difference3}  \\ \hline
			\texttt{param.maxit} & $200$ & $500$ \\
			\texttt{param.tol} & $10^{-3}$ & $10^{-5}$
		\end{tabular}
	\end{center}

	In Table~\ref{tab:difference2}, the results of synthesis-based model using the DR algorithm are also recomputed with the new choice of parameters.
	Note that the inconsistent approach is omitted, since even with the less strict convergence criteria,
	the results of the two approaches are almost identical.
	For the purpose of direct comparison, Table~\ref{tab:difference3} presents the differences in $\snr$ shown in Table~\ref{tab:difference2}.
	As mentioned above, positive values indicate a~better performance of the implementation based on $\prox_{\sparse\circ\ana}$.
	
	\begin{table}[h]
		\hspace{0.05\linewidth}
		\parbox{0.5\linewidth}{\begingroup\small
			\caption{Values for the DR (original) and the CP (new) algorithms for the analysis-based approach.
			The parameters were set to \texttt{param.maxit = 500} and \texttt{param.tol = 1e-5}.}
			\label{tab:difference2}
			\begin{tabular}{|l|rrr|rrr|}
				\hline
				\# & \multicolumn{3}{l|}{DR} & \multicolumn{3}{l|}{CP} \\ \cline{2-7} 
				& GAB & WAV & ERB & GAB & WAV & ERB \\ \hline
				1 & $18.01$ & $26.44$ & $26.79$ & $18.53$ & $26.61$ & $26.98$ \\
				2 & $19.34$ & $26.23$ & $26.28$ & $19.99$ & $26.34$ & $26.42$ \\
				3 & $18.39$ & $19.28$ & $19.37$ & $18.58$ & $19.24$ & $19.36$ \\
				4 & $15.84$ & $19.77$ & $20.45$ & $16.16$ & $19.76$ & $20.42$ \\ \hline
			\end{tabular}\endgroup
		}
		\hfill
		\parbox{0.3\linewidth}{\begingroup\small
			\caption{Difference in the values of $\snr$ in dB presented in Table~\ref{tab:difference2}. \\}
			\label{tab:difference3}	
			\begin{tabular}{|rrr|}
				\hline
				\multicolumn{3}{|l|}{DR/CP} \\ \cline{1-3} 
				GAB & WAV & ERB \\ \hline
				$0.52$ & $ 0.16$ & $ 0.19$ \\
				$0.65$ & $ 0.11$ & $ 0.14$ \\
				$0.19$ & $-0.04$ & $-0.01$ \\
				$0.33$ & $-0.01$ & $-0.03$ \\ \hline
			\end{tabular}\endgroup
		}
		\hspace{0.05\linewidth}
	\end{table}
	
	It is clear from Table~\ref{tab:difference3} that letting the algorithms converge closer to the actual solution of the corresponding optimization tasks reduces the difference between the more and the less accurate approach.
	Nonetheless, it can still be seen that not only the results in general but also the inaccuracy of the DR algorithm in the analysis-based approach depend
	on the choice of (redundant) time-frequency representation of the audio signal.

\subsection{Software and reproducible research}
	\label{sec:software}

	All the experimental data were generated using MATLAB R2017b while using LTFAT \cite{LTFAT} version 2.3.1 and NSGToolbox \cite{Balazs2011:nonstatgab} version 0.1.0.
	\edit{The example in \ref{appendix:example} is generated using CVX, a package for specifying and solving convex programs \cite{CVX,CVX.paper}.}
	The MATLAB codes are available online at \url{https://github.com/ondrejmokry/ApproximalOperator}.

\section{Conclusions}

	The article presents the proximal operator of a composition of a~convex function with a~linear mapping $\ana$ that satisfies $\syn\!\ana = \alpha\Id$,
	which is the case of $\ana$ being the (redundant) analysis operator of a tight frame.
	So far, a compact formula for such an operator has been known only for the composition with $\syn$ (in the current notation).
	The theoretical derivation does not yield a~really effective way to apply the proximal operator in practice.
	Nevertheless, it is shown that using the fast explicit approximation from \cite{Lieb2018:Audio.Inpainting}, the \emph{approximal operator},
	is not only meaningful, but it is also reasonably close to the proper solution obtained with the exact proximal operator.

	The practical experiment was limited only to audio inpainting scenarios originally performed in \cite{Lieb2018:Audio.Inpainting},
	yet the theoretical result has a~straightforward application to the problems of audio declipping or dequantization as well, or even in the related problems in image and video processing.

\appendix

\section{Proofs}
\label{appendix:proofs}

\begin{proof}[Proof of Lemma~\ref{lemma:the.new}]
	\edit{Recall that $g\colon\RR^m\to\RR\cup\{\infty\}$ is assumed to be a~proper convex function, $f(\x) = g(\ana\x+\b)$, where $\b\in\RR^m$ and $\ana\colon\mathbb{V}\to\RR^m$ is a~linear transformation satisfying $\syn\!\ana=\alpha\Id$ for some constant $\alpha>0$.}
	We start the proof by quoting the first part of the proof of Lemma \ref{lemma:the.original},
	as presented in \cite[pp.~140--141]{Beck2017:First.Order.Methods},
	since the assumed property of the linear transform is not crucial at first.
	
	By definition, $\prox_f\colon\mathbb{V}\to\mathbb{V}$ for $f(\u) = g(\ana\u+\b)$ is a mapping such that $\prox_f(\x)$ is the optimal solution to
	\begin{equation}
	\min _{\u \in \mathbb{V}}\left\{g(\ana\u+\b)+\frac{1}{2}\norm{\u-\x}^{2}\right\}.
	\end{equation}
	The above problem can be formulated as the following constrained problem:
	\begin{equation}
	\begin{array}{ll}{\min _{\u \in \mathbb{V}, \z \in \RR^{m}}} & {\left\{g(\z)+\frac{1}{2}\norm{\u-\x}^{2}\right\}} \\[0.5em] {\text { s.\,t. }} & {\z=\ana\u+\b.}\end{array}
	\label{eq:const.prob}
	\end{equation}
	Denote the optimal solution of \eqref{eq:const.prob} by $(\tilde{\u},\tilde{\z})$ (the existence and uniqueness of $\tilde{\u}$ and $\tilde{\z}$ follow from the underlying assumption that $g$ is proper close and convex).
	Note that $\tilde{\u}=\prox_f(\x)$.
	Fixing $\z=\tilde{\z}$, we obtain that $\tilde{\u}$ is the optimal solution to
	\begin{equation}
	\begin{array}{ll}{\min _{\u \in \mathbb{V}}} & {\frac{1}{2}\norm{\u-\x}^{2}} \\[0.5em] {\text { s.\,t.}} & {\ana\u=\tilde{\z}-\b.}\end{array}
	\label{eq:const.prob.strongly.dual}
	\end{equation}
	Since strong duality holds for the problem \eqref{eq:const.prob.strongly.dual}
	\cite[pp.\,439--440]{Beck2017:First.Order.Methods}, it follows that there exists $\y\in\RR^m$ for which the two conditions
	\begin{align}
	\tilde{\u} & \in \argmin_{\u \in \mathbb{V}}\left\{\frac{1}{2}\norm{\u-\x}^{2}+\langle\y, \ana\u-\tilde{\z}+\b\rangle\right\},
	\label{eq:aligned.first.row}\\
	\ana\tilde{\u} &=\tilde{\z}-\b \label{eq:aligned.second.row}
	\end{align}
	are satisfied.
	Since the objective function in \eqref{eq:aligned.first.row} is strictly convex and differentiable,
	the unique minimizer $\tilde{\u}$ is obtained by setting its gradient to zero, which leads to
	\begin{equation}
	\tilde{\u} = \x - \syn\y.
	\label{eq:tilde.u}
	\end{equation}
	Substituting this expression of $\tilde{\u}$ into \eqref{eq:aligned.second.row}, we obtain
	\begin{equation}
	\ana(\x-\syn\y) = \tilde{\z} - \b.
	\label{eq:important.relation}
	\end{equation}
	Quoting \cite{Beck2017:First.Order.Methods} must be stopped here, since we will further utilize a~different assumption on the linear operator.
	
	Since $\syn\!\ana=\alpha\Id$, applying the synthesis $\syn$ onto both sides of Eq.\,\eqref{eq:important.relation} leads to
	\begin{equation}
	\x-\syn\y = \alpha^{-1}\syn(\tilde{\z}-\b).
	\label{eq:equivalant.stuff}
	\end{equation}
	The important observation here is that \eqref{eq:equivalant.stuff} is equivalent to \eqref{eq:important.relation} only when $\tilde{\z}-\b\in\mathcal{R}(\ana)$, which is enforced by the left hand side of the relation \eqref{eq:important.relation}.
	Substituting this result into \eqref{eq:tilde.u} leads to an explicit expression for $\tilde{\u}$ in terms of $\tilde{\z}$ (still limited to $\tilde{\z}-\b\in\mathcal{R}(\ana)$):
	\begin{equation}
	\tilde{\u} = \alpha^{-1}\syn(\tilde{\z}-\b).
	\label{eq:tilde.u.2}
	\end{equation}
	Plugging this result in the minimization problem \eqref{eq:const.prob}, we obtain that $\tilde{\z}$ is given by
	\begin{equation}
	\tilde{\z} = \argmin_{\z}\left\{g(\z) + \irab(\z) + \frac{1}{2}\left\|\alpha^{-1}\syn(\z-\b)-\x\right\|^2 \right\},
	\label{eq:not.great.not.terrible}
	\end{equation}
	where the relation $\tilde{\z}-\b\in\mathcal{R}(\ana)$ is enforced by the indicator function $\irab$.
	
	Now recall two useful properties:
	\begin{enumerate}
		\item Viewing $\ana$ from the perspective of frame theory,
		it is an analysis operator corresponding to a tight frame, for which it holds
		$\norm{\ana\z}^2 = \alpha\norm{\z}^2$ for all $\z\in\mathbb{V}$.
		\item Since we require $\tilde{\z}-\b\in\mathcal{R}(\ana)$, it follows from Eq.\,\eqref{eq:proj.RA} that $\alpha^{-1}\ana\syn(\tilde{\z}-\b) = \tilde{\z}-\b$.
	\end{enumerate}
	
	Using the first property, we can rewrite \eqref{eq:not.great.not.terrible} as
	\begin{equation}
	\tilde{\z} = \argmin_{\z} \left\{ g(\z) + \irab(\z) + \frac{1}{2\alpha}\left\|\ana\alpha^{-1}\syn(\z-\b)-\ana\x\right\|^2 \right\}.
	\end{equation}
	Using the second property and multiplying the objective function by the (positive) constant $\alpha$ leads to
	\begin{align}
	\tilde{\z} &= \argmin_{\z}\left\{ \alpha g(\z) + \irab(\z) + \frac{1}{2}\left\|\z-(\ana\x+\b)\right\|^2 \right\}\\
	&=\prox_{\alpha g+\irab}(\ana\x+\b).\label{eq:z.is.prox}
	\end{align}
	Plugging the expression for $\tilde{\z}$ into \eqref{eq:tilde.u.2} produces the desired result.
\end{proof}

\begin{proof}[Proof of Lemma~\ref{lemma:approx.is.prox}]
	To prove that a function $F\colon\mathbb{V}\to\mathbb{V}$ is a proximal operator of a convex lower semi-continuous function, it is sufficient to show two properties \cite[Corollary 10.c]{Moreau1965:Proximite.dualite}:
	\begin{enumerate}
		\item there exists a convex lower semi-continuous function $\psi$ such that for any $\y\in\mathbb{V}$, $F(\y)\in\partial\psi(\y)$,
		\item $F$ is non-expansive, i.e.
		\begin{equation}
		\left\|F(\y)-F\left(\y^{\prime}\right)\right\| \leq\left\|\y-\y^{\prime}\right\|, \quad \forall \y, \y^{\prime} \in \mathbb{V}.
		\label{eq:non.expansive}
		\end{equation}
	\end{enumerate}
	The symbol $\partial\psi(\y)$ denotes the subdifferential of the function $\psi$ at the point $\y$,
	see e.g.\ \cite[Def.\,1.2.1, p.\,167]{Hiriart-UrrutyLemarechal2001:Fundament.convex.analysis}.

	Observe that the first property is necessarily satisfied for $F = \prox_{\alpha g}$, since it is a proximal operator.
	This means that there exists a convex lower semi-continuous function $\eta$ such that
	\begin{equation}
	\forall\x\in\RR^m\quad\prox_{\alpha g}(\x)\in\partial\eta(\x) = \{ \vect{s} \mid \eta(\y)\geq \eta(\x) + \langle \vect{s}, \y-\x \rangle\;\forall\y\in\RR^m \}.
	\end{equation}
	If this relation holds for all $\x$, it must also hold for vectors in the form $\x=\ana\z$, leading to
	\begin{equation}
	\forall\z\in\mathbb{V}\quad\prox_{\alpha g}(\ana\z)\in\partial\eta(\ana\z) = \{ \vect{s} \mid \eta(\y)\geq \eta(\ana\z) + \langle \vect{s}, \y-\ana\z \rangle\;\forall\y\in\RR^m \}.
	\end{equation}
	Then also	
	\begin{equation}
	\forall\z\quad\alpha^{-1}\syn\prox_{\alpha g}(\ana\z)\in\alpha^{-1}\syn\partial\eta(\ana\z).
	\label{eq:subdifferential}
	\end{equation}
	Recalling the property \cite[Theorem 4.2.1, p.\,184]{Hiriart-UrrutyLemarechal2001:Fundament.convex.analysis}
	\begin{equation}
		\partial(h\circ A)(\x) = \transp{A}\partial h(A\x)\quad\forall\x\in\RR^n,
	\end{equation}
	\edit{Eq.\,\eqref{eq:subdifferential} further implies}
	\begin{equation}
	\forall\z\quad\alpha^{-1}\syn\prox_{\alpha g}(\ana\z)\in\alpha^{-1}\partial(\eta\circ\ana)(\z) = \partial(\alpha^{-1}\eta\circ\ana)(\z).
	\end{equation}
	Since $\eta$ is a lower semi-continuous and convex function and $\ana$ is a linear operator, $\eta\circ\ana$ and $\alpha^{-1}\eta\circ\ana$ are also lower semi-continuous and convex.
	Property 1. is thus satisfied for $F = \badprox_{f}$ with $\psi = \alpha^{-1}\eta\circ\ana$.
	
	The non-expansivity can be shown similarly:
		Substituting \eqref{eq:shortened.bad.prox} into \eqref{eq:non.expansive} and using the fundamental property of operator norm leads to
	\begin{equation}
	\norm{\alpha^{-1}\syn\left(\prox_{\alpha g}(\ana\y)-\prox_{\alpha g}(\ana\y^\prime)\right)} \leq \norm{\alpha^{-1}\syn}\norm{\prox_{\alpha g}(\ana\y)-\prox_{\alpha g}(\ana\y^\prime)}.
	\label{eq:proof.first}
	\end{equation}
	
	First, let us compute the operator norm of $\norm{\alpha^{-1}\syn}$.
	The property $\norm{\ana\x}^2 = \alpha\norm{\x}^2\;\forall\x$ implies $\norm{\ana\x} = \sqrt{\alpha}\norm{\x}$, i.e. $\norm{\ana} = \sqrt{\alpha}$.
	Thus
	\begin{equation}
	\norm{\alpha^{-1}\syn} = \alpha^{-1}\norm{\syn} = \alpha^{-1}\norm{\ana} = \frac{1}{\sqrt{\alpha}}.
	\label{eq:proof.second}
	\end{equation}
	
	Now, since $\prox_{\alpha g}$ meets \eqref{eq:non.expansive}, it holds
	\begin{equation}
	\norm{\prox_{\alpha g}(\ana\y)-\prox_{\alpha g}(\ana\y^\prime)} \leq \norm{\ana\y-\ana\y^\prime} = \norm{\ana(\y-\y^\prime)} = \sqrt{\alpha}\norm{\y-\y^\prime}.
	\label{eq:proof.third}
	\end{equation}
	Plugging \eqref{eq:proof.second} and \eqref{eq:proof.third} into \eqref{eq:proof.first} shows that \eqref{eq:non.expansive} truly holds for the $\badprox_{f}$ operator.
\end{proof}

\edit{
\section{Illustrative example}
	\label{appendix:example}
	The example illustrates the difference of $f$ and $\phi$, which are the functions that $\prox_{f}$ and $\prox_\phi=\badprox_{f}$ belong to, respectively.
	The example is closely related to the experiment in Sec.\,\ref{sec:experiments}.
	The function $g$ is the $\ell_1$ norm on $\RR^m=\RR^4$, $g=\norm{\cdot}_1\colon\RR^4\to\RR$, and $f=g\circ\ana\colon\mathbb{V}=\RR^n=\RR^2\to\RR$.
	The dimensions are set such that the functions $f$ and $\phi$ can be visualized.
	The operator $\ana\colon\RR^2\to\RR^4$ is identified with the matrix
	\begin{equation}
	\matr{A} = \begin{bmatrix*}[r]
	0.7464 & 0.0444 \\
	0.1588 & 0.9127 \\
	-0.9348 & 0.7795 \\
	-0.7375 & -0.7466
	\end{bmatrix*}
	\end{equation}
	and represents the analysis of a tight frame with $\alpha = 2$.
	The synthesis operator can be represented by the transposed matrix
	\begin{equation}
	\transp{\matr{A}} = \begin{bmatrix*}[r]
	0.7464 & 0.1588 & -0.9348 & -0.7375 \\
	0.0444 &   0.9127 & 0.7795 & -0.7466
	\end{bmatrix*}.
	\end{equation}
	Fig.\,\ref{fig:example:generators} depicts the columns of $\transp{\matr{A}}$, i.e.\ the tight frame as a set of vectors (generators) in the real plane.
}

\edit{	
	In the finite-dimensional and convex setting, the function $\phi$ can be numerically computed using Eq.\,\eqref{eq:phi} as
	\begin{equation}
		\phi(\x) = \alpha^{-1}\min_{\y\in\RR^4}\left\{ \alpha\norm{\y}_1 + \frac{1}{2}\norm{\ana\x-\y}^2 + \iota_{\mathcal{N}(\syn)}(\ana\x-\y) \right\}.
	\end{equation}
}%
\edit{%
	For the purpose of this illustration, we sample the functions $f$ and $\phi$ on a discrete grid in $\RR^2$.
	To be more specific, the values for $\x\in\RR^2$ are taken equidistantly from the range $[-2,2]\times[-2,2]$ with a step of $0.2$ in each direction.
}%
\edit{%
	The result is plotted in Fig.\,\ref{fig:example:functions}.
	As expected based on Lemma~\ref{lemma:leq}, the function $\phi$ is less or equal to $f$ at each point.
	Furthermore, the functions $f$ and $\phi$ closely resemble each other, which explains the results in the audio inpainting experiment.}
\edit{
	\begin{figure}[t]
		\centering
		\begin{subfloat}[\edit{the tight frame}]{
				\includegraphics[width=0.4\textwidth]{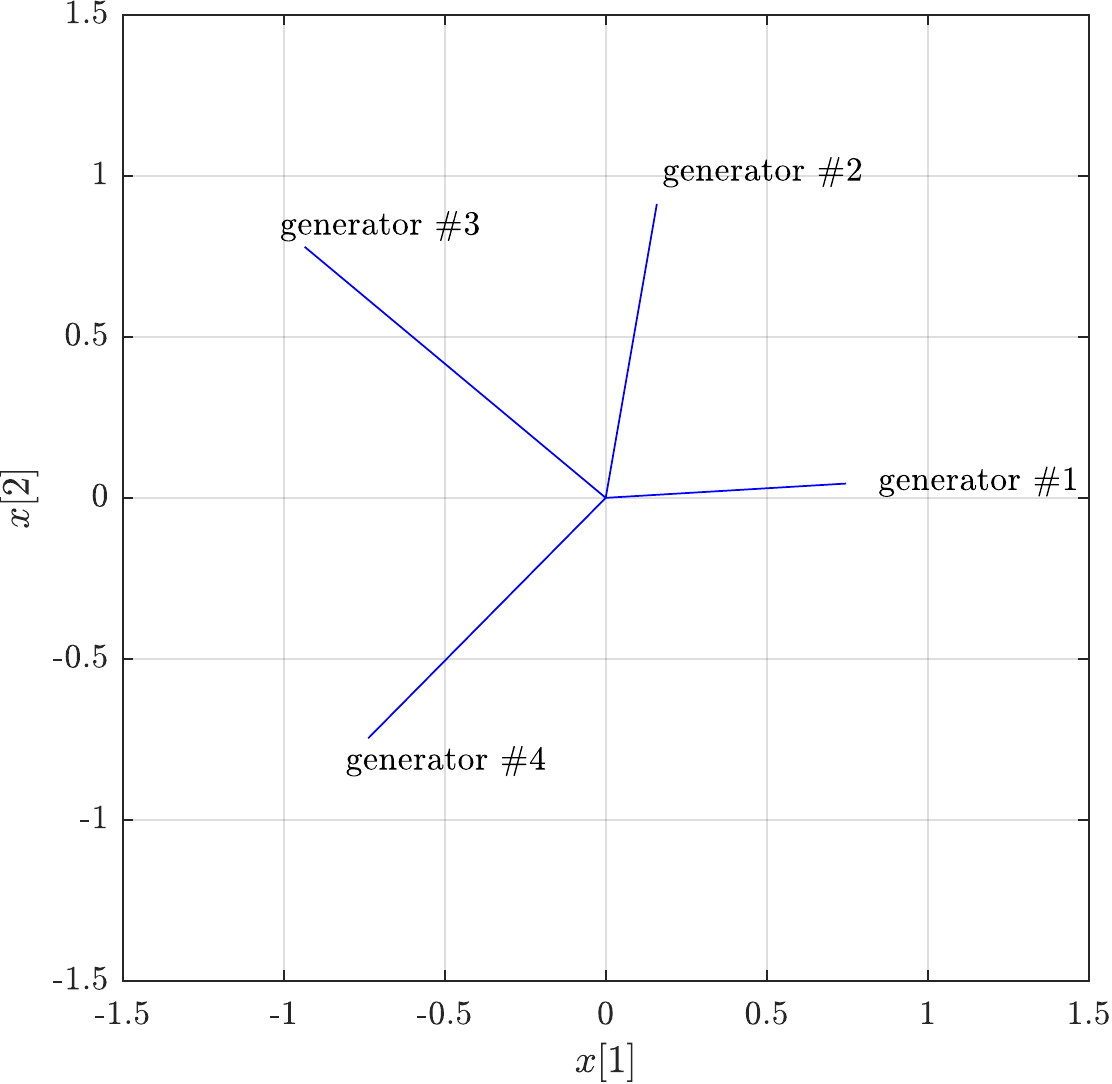}
				\label{fig:example:generators}}
		\end{subfloat}
		\hfill
		\begin{subfloat}[\edit{comparison of the functions $f$ and $\phi$}]{
				\includegraphics[width=0.5\textwidth]{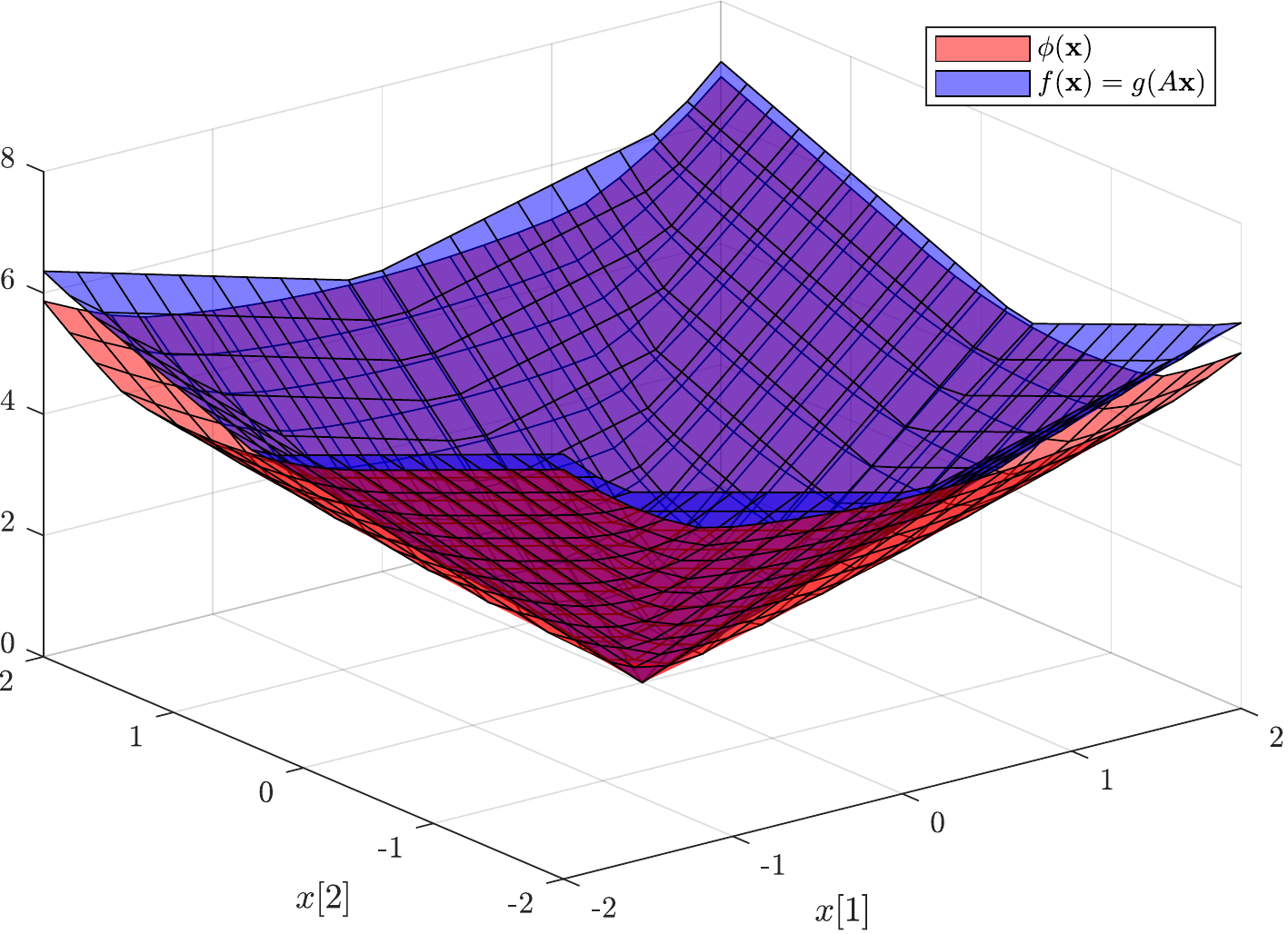}
				\label{fig:example:functions}}
		\end{subfloat}
		\caption{\edit{Illustrative example.
			Proximal operator of the function $f=g\circ\ana$ is defined by Lemma \ref{lemma:the.new}.
			When this operator is only approximated by $\badprox_f$, the underlying functional of the proximal operator is $\phi$, defined in Lemma~\ref{lemma:phi}.
			The example is plotted for $g=\norm{\cdot}_1$ and $\syn$ consisting as a matrix of the generators in Fig.\,\ref{fig:example:generators}.}
		}
		\label{fig:example}
	\end{figure}
}

\section{Comments on the use of complex-valued operators}
\label{appendix:complex}

	In the experimental part, the complex-valued analysis operator $\ana\colon\CC^n\to\CC^m$ and its synthesis counterpart $\adjoint{\ana}$ (instead of $\syn$) are used.
	Nonetheless, both Lemma~\ref{lemma:the.original} and Lemma~\ref{lemma:the.new} remain valid also in such a case.
	The place in the proof that could make trouble in the complex case is the formulation of the Lagrangian in Eq.\,\eqref{eq:aligned.first.row}, since it should be real;
	the rest of the manipulations also hold in the complex case.
	
	To deal with the Lagrangian for complex variables, observe that
	\begin{equation}
		\ana\u = \tilde{\z}-\b \ \Leftrightarrow\  \Re(\ana\u-\tilde{\z}+\b)=\mathbf{0}\wedge \Im(\ana\u-\tilde{\z}+\b)=\mathbf{0},
	\end{equation}
	which can be rewritten in matrix form as
	\begin{equation}
	\underbrace{\left[\begin{array}{cc}
		{\Re(\ana)} & {-\Im(\ana)} \\
		{\Im(\ana)} & {\Re(\ana)}
		\end{array}\right]}_{\hat{\ana}} \underbrace{\left[\begin{array}{c}
		{\Re(\u)} \\
		{\Im(\u)}
		\end{array}\right]}_{\hat{\u}}-\underbrace{\left[\begin{array}{c}
		{\Re(\tilde{\z}-\b)} \\
		{\Im(\tilde{\z}-\b)}
		\end{array}\right]}_{\hat{\c}}=\mathbf{0},
	\end{equation}
	where $\hat{\ana}\colon\RR^{2n}\to\RR^{2m}$, $\hat{\u}\in\RR^{2n}$ and $\hat{\c}\in\RR^{2m}$.
	Rewriting $\x\in\CC^n$ in a similar way into a real vector $\hat{\x}\in\RR^{2n}$ leads to the Lagrangian
	\begin{equation}
		\frac{1}{2}\norm{\hat{\u}-\hat{\x}}^{2}+\langle\hat{\y}, \hat{\ana}\hat{\u}-\hat{\c}\rangle,
		\label{eq:lagrangian.real}
	\end{equation}
	where $\hat{\y}\in\RR^{2m}$.
	The unique minimizer of \eqref{eq:lagrangian.real} is
	\begin{equation}
		\hat{\u} = \hat{\x} - \hat{\ana}^\top\hat{\y}.
	\end{equation}
	It can be easily shown that the real and imaginary parts of the complex solution $\tilde{\u} = \x - \adjoint{\ana}\y$
	obtained using Eq.\,\eqref{eq:tilde.u} match the results in just presented real case.
	
	Note that in the case of the complex Gabor transform presented earlier,
	the mapping $\prox_f$ should map the real signal $\x$ to a real signal $\tilde{\u}$.
	This would be ensured by a particular convex-conjugate structure of $\ana$, $\adjoint{\ana}$ and also $\y$ in such a case.
	For a~more detailed description, see the discussion in \cite[p.\,9]{RajmicZaviskaVeselyMokry2019:Axioms} and the references therein.
	
	A~similar argumentation would be used regarding the proof of Lemma \ref{lemma:approx.is.prox} and the concept of subdifferentials,
	where, once again, we would define the real-valued inner product for complex vectors.

	\newcommand{\noopsort}[1]{} \newcommand{\printfirst}[2]{#1}
	\newcommand{\singleletter}[1]{#1} \newcommand{\switchargs}[2]{#2#1}

\end{document}

%% file: defs.tex
\def\RR{\mathbb{R}}

\def\CC{\mathbb{C}}

\theoremstyle{plain}

\newtheorem{lemma}{Lemma}
\theoremstyle{definition}

\theoremstyle{remark}

\DeclareTextAccent{\ring}{OT1}{23}

\def\x{\vect{x}}
\def\u{\vect{u}}

\def\y{\vect{y}}
\def\z{\vect{z}}

\def\b{\vect{b}}

\def\c{\vect{c}}

\def\snr{\mathrm{SNR}}

\newcommand{\Id}{\mathrm{Id}}
\newcommand{\syn}{\transp{\ana}}
\newcommand{\ana}{A}
\newcommand{\RA}{\mathcal{R}(\ana)}

\newcommand{\sig}{\y}
\newcommand{\mask}{M}
\newcommand{\sparse}{\mathcal{S}}

\newcommand{\norm}[1]{\|#1\|}

\newcommand{\transp}[1]{#1^\top}
\newcommand{\adjoint}[1]{#1^*}
\newcommand{\inv}[1]{#1^{-1}}
\newcommand{\pinv}[1]{#1^{\dagger}}

\newcommand{\vect}[1]{\mathbf{#1}} 
\newcommand{\matr}[1]{\mathbf{#1}} 

\newcommand{\argmin}{\mathop{\operatorname{arg~min}}}

\newcommand{\prox}{\mathrm{prox}}
\newcommand{\badprox}{\mathrm{approx}}

\newcommand{\proj}{\mathrm{proj}}

\newcommand{\rank}[1]{\mathrm{rank}(#1)}
\newcommand{\conv}{\mathop{\square}}

\newcommand{\edit}[1]{\textcolor{black}{#1}}

\newcommand{\irab}{\iota_{(\mathcal{R}(\ana)+\b)}}